\documentclass[letterpaper,twocolumn,10pt]{article}
\usepackage{usenix_2020}

\usepackage{cite}
\usepackage{amsmath,amssymb,amsfonts,amsthm}
\usepackage[linesnumbered,ruled,vlined]{algorithm2e}
\usepackage{graphicx}
\usepackage{textcomp}
\usepackage{xcolor}

\usepackage{tikz}
\usetikzlibrary{patterns, arrows, positioning, calc}
\usepackage{comment}

\usepackage{url}
\usepackage{xspace}
\usepackage{enumitem}
\usepackage{fontawesome5}

\usepackage{subcaption}
\usepackage{makecell}
\usepackage{multirow}
\usepackage{booktabs}
\usepackage{bbding}
\usepackage{tablefootnote}
\usepackage{listings}
\usepackage{pifont}
\usepackage{array}
\usepackage{balance}
\usepackage{setspace}

\definecolor{keywordred}{RGB}{180,0,0}
\definecolor{funcgreen}{RGB}{0,140,100}
\definecolor{typeyellow}{RGB}{180,140,0}
\definecolor{codecomment}{RGB}{100,100,100}

\lstdefinestyle{PythonStyle}{
    language=Python,
    basicstyle=\small\ttfamily,
    keywordstyle=\color{keywordred}\bfseries,
    stringstyle=\color{typeyellow},
    commentstyle=\color{codecomment}\itshape,
    emph={Dispatch,OnComplete,Cancel,QueryEnsemble},
    emphstyle=\color{funcgreen},
    emph={[2]str,int,bool,ConsensusCtx,RequestHandle,Callable,Answer,EnsembleState},
    emphstyle={[2]\color{typeyellow}},
    showstringspaces=false,
    breaklines=true,
    frame=none,
    xleftmargin=1em,
    aboveskip=0.5em,
    belowskip=0.5em,
}

\def\Snospace~{\S{}}

\newtheorem{theorem}{Theorem}
\newtheorem{lemma}{Lemma}
\theoremstyle{definition}
\newtheorem{definition}{Definition}
\newtheorem{assumption}{Assumption}

\newcommand{\para}[1]{\noindent \textbf{#1 }}

\definecolor{darkgreen}{rgb}{0.078,0.667,0.016}

\newcommand{\ie}{i.e.\@\xspace}

\newcommand{\sys}{Aegean\xspace}
\newcommand{\mtype}[1]{\texttt{#1}\xspace}
\newcommand{\state}[1]{\texttt{#1}\xspace}
\newcommand{\msg}[2]{$\langle$\mtype{#1}, \textit{#2}$\rangle$\xspace}

\begin{document}


\title{Reaching Agreement Among Reasoning LLM Agents}

\author{
{\rm Chaoyi Ruan}\\
NUS
\and
{\rm Yiliang Wang}\\
NUS
\and
{\rm Ziji Shi}\\
NUS
\and
{\rm Jialin Li}\\
NUS
} 

\maketitle

\begin{abstract}
\noindent
Multi-agent systems have extended the capability of agentic AI.
Instead of single inference passes, multiple agents perform collective reasoning to derive high quality answers.
However, existing multi-agent orchestration relies on static heuristic workflows such as fixed loop limits and barrier synchronization.
These ad-hoc approaches waste computational resources, incur high latency due to stragglers, and risk finalizing transient agreements.
We argue that reliable multi-agent reasoning requires a formal foundation analogous to classical distributed consensus problem.

To that end, we propose a formal model of the multi-agent refinement problem.
The model includes definitions of the correctness guarantees and formal semantics of agent reasoning.
We then introduce \sys, a consensus protocol designed for stochastic reasoning agents that solves multi-agent refinement.
We implement the protocol in \sys-Serve, a consensus-aware serving engine that performs incremental quorum detection across concurrent agent executions, enabling early termination when sufficient agents converge.
Evaluation using four mathematical reasoning benchmarks shows that \sys provides provable safety and liveness guarantees while reducing latency by 1.2--20$\times$ compared to state-of-the-art baselines, maintaining answer quality within 2.5\%.
Consistent gains across both local GPU deployments and commercial API providers validate that consensus-based orchestration eliminates straggler delays without sacrificing correctness.

\end{abstract}

\section{Introduction}
\label{sec:intro}

\noindent
A new computational pattern is reshaping how AI systems solve complex problems: Rather than relying on a single inference pass, modern systems orchestrate multiple reasoning paths that deliberate collectively to produce answers that exceed the capabilities of individual attempts.
This pattern emerges in frontier models like Gemini Deep Research~\cite{google_deepthink_2025} and Grok Heavy~\cite{xai_grok4_2025} that explore diverse solution trajectories to multi-agent architectures where independent model instances engage in iterative discussion~\cite{du2023mad,wu2023autogen,li2023camel,han2024multiagent}, critique each other's reasoning~\cite{chai2025scimaster,chen2025tumix}, and progressively refine their answers~\cite{madaan2023self,li2025selfmoa,zhang2024dei}.
These approaches share a common structure: Parallel reasoning processes generate candidate solutions, exchange information, and eventually converge on a single result.
They face a challenge similar to the core problem in distributed consensus: How should a set of independent processes reach agreement over conflicting proposals?

\begin{figure}[!t]
 \centering
 \includegraphics[width=0.45\textwidth]{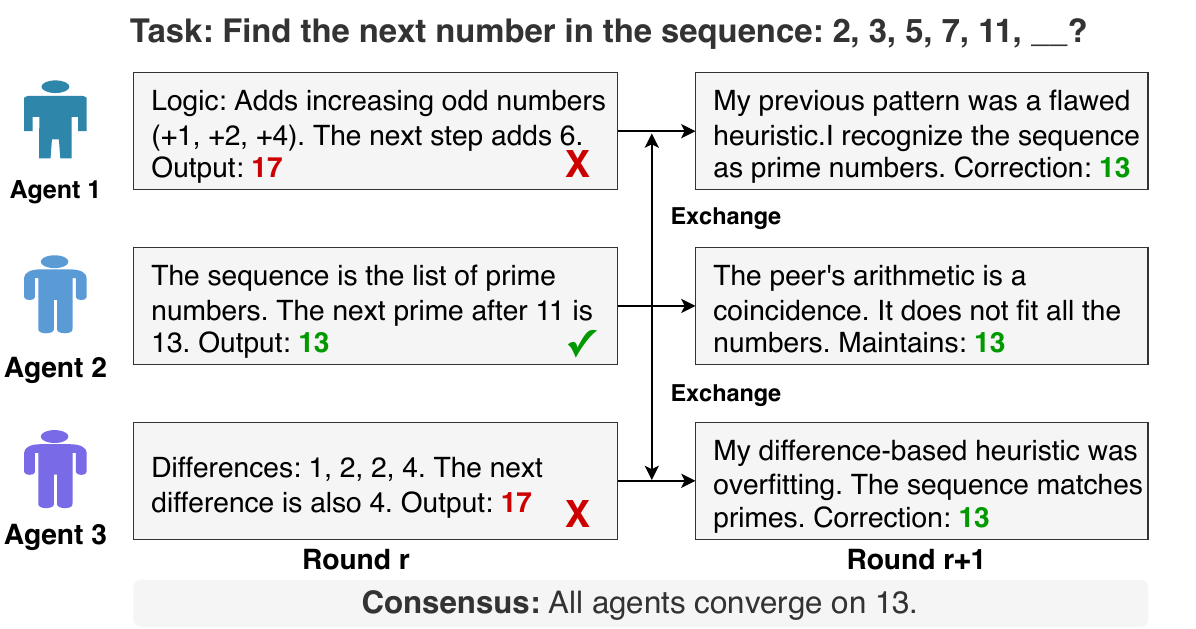}
 \caption{A multi-agent consensus example from 3 models.}
 \label{fig:example}
\end{figure}

\autoref{fig:example} illustrates this paradigm: Initially, Agent~1 and Agent~3 propose 17 
based on flawed arithmetic, while Agent~2 correctly identifies the prime sequence and outputs 13.
After exchanging reasoning traces, Agent~1 and Agent~3 recognize their error and converge to the correct answer.
This \textit{refinement process} enables higher solution accuracy than single inference passes.
However, it is unclear \textit{when} the system should output a solution and \textit{which} solution it should output.
Outputting a solution with majority support in round $r$ would capture the wrong answer;
delaying the output in rounds after $r+1$ would unnecessarily prolong the refinement process.

Classical distributed consensus problem addresses similar questions.
Unfortunately, consensus protocols such as Paxos~\cite{lamport2001paxos} and Raft~\cite{ongaro2014raft} do not apply in the multi-agent refinement context.
These protocols assume deterministic state machines.
LLMs, however, are inherently stochastic where the same query may lead to different responses.
Consensus protocols apply quorum intersection to ensure committed values are stable.
A solution supported by majority agents in a round may dissolve in the next round as agents refine their positions after considering peer rationales.
Lastly, consensus protocols do not enforce constraints on the committed value.
Multi-agent refinement aims to derive high-quality answers.

Fundamentally, existing literature~\cite{madaan2023self,li2025selfmoa,zhang2024dei,du2023mad,chai2025scimaster,chen2025tumix} lacks a formal model of agreement among stochastic LLM agents.
The model should define the semantics of agent reasoning, answer refinement, and termination conditions.
Without such a model, existing multi-agent solutions compromise on accuracy and performance.
First, heuristic termination mechanisms decouple computational cost from actual convergence.
The coordinators rely on fixed round limits, wasting compute on easy queries while prematurely halting complex ones.
Second, barrier-based orchestration suffers high completion time due to tail stragglers.
Standard workflows require all reasoning paths to complete before evaluating agreement~\cite{wang2025mix,chen2024magicore}, causing the slowest agent to dictate latency.
\autoref{tab:convergence-latency} shows that the slowest model has 11-64$\times$ higher latency than the fastest model.
Lastly, systems that finalize upon observing simple majority cannot distinguish unstable states from stable agreement.
They risk committing transient states that would dissolve in the next round. 

Addressing these challenges requires rethinking the foundations underlying multi-agent coordination. 
In this work, we formalize the multi-agent refinement problem.
We define the safety and liveness properties of multi-agent refinement, and provide formal semantics of agent reasoning.
To validate our model, we conducted experiments on a diverse set of real-world AI models.
The results provide strong empirical support for our agent refinement semantics.

We present \sys, a consensus framework that bridges non-deterministic reasoning and deterministic system guarantees.
Based on our formal model, we design an \sys protocol that applies leader-based coordination to progress rounds of agent refinement.
The protocol guarantees safety and liveness, including \textit{refinement monotonicity} where the quality of solutions strictly improves over rounds.
It also incorporates a modular refinement decision engine which defines customizable agreement conditions, such as solution similarity threshold and stability horizon.
Built atop this protocol, \sys-Serve embeds consensus coordination directly within inference infrastructure, enabling progressive quorum detection that evaluates outputs as they stream.
The engine immediately cancels in-progress reasoning paths once quorum is secured, decoupling latency from stragglers.
Evaluation across four reasoning benchmarks (GSM8K~\cite{cobbe2021gsm8k}, MMLU~\cite{hendryckstest2021,hendrycks2021ethics}, AIME~\cite{aime2024}, IMO~\cite{luong-etal-2025-towards}) demonstrates that \sys reduces average latency by 1.2--20$\times$ and P99 tail latency by up to 11$\times$ compared to state-of-the-art baselines.
The system maintains accuracy within 2.5\% while reducing token consumption by 1.1--4.4$\times$ through early termination when consensus stabilizes.
These gains hold consistently across locally deployed models and commercial model API, confirming that consensus-based orchestration eliminates straggler delays without sacrificing correctness.

We make the following contributions:
\begin{itemize}[noitemsep, topsep=0pt]
\item We formalize the first consensus model for stochastic multi-agent LLMs, establishing round-based stability as the necessary condition for safe commitment.

\item We design the first consensus protocol, \sys, for non-deterministic participants, introducing stability horizons and quorum-based finalization with provable safety and liveness guarantees.

\item We build the first consensus-aware serving engine, \sys-Serve, enabling consensus-aware execution and achieving significant performance benefits.
\end{itemize}
\section{Background}
\label{sec:background}
\label{subsec:multi-agent-llm}

\begin{figure}[!t]
    \centering
    \begin{subfigure}[b]{0.47\textwidth}
        \centering
        \includegraphics[width=\linewidth]{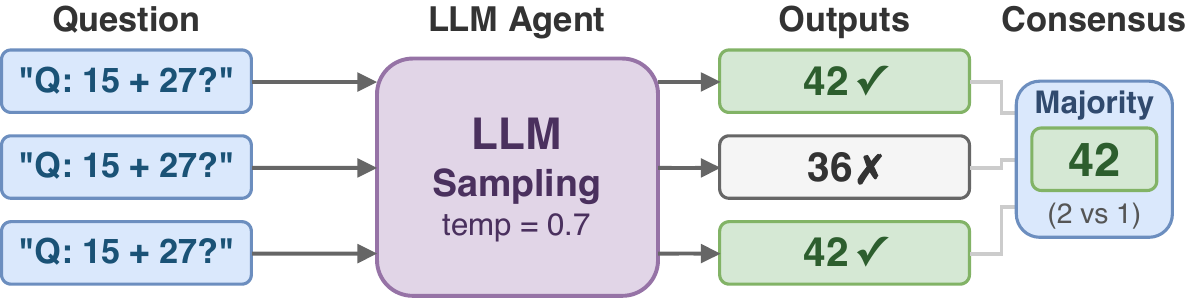} 
        \caption{Single-agent consensus process}
        \label{fig:sampling}
    \end{subfigure}
    \par\bigskip
    \begin{subfigure}[b]{0.47\textwidth}
        \centering
        \includegraphics[width=\linewidth]{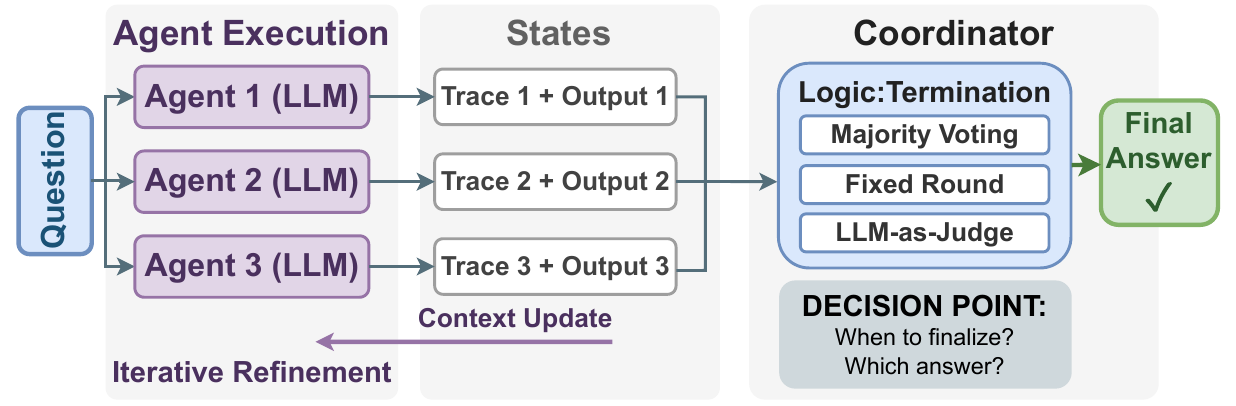}
        \caption{Multi-agent consensus process}
        \label{fig:system_arch}
    \end{subfigure}
    
    \caption{Architecture of agentic consensus. (a) Single-agent consensus aggregates samples from one model; (b) multi-agent consensus coordinates multiple agents through iterative refinement with a central coordinator.}
    \label{fig:architecture}
\end{figure}

\para{From single response to collective consensus.}
Large language models generate outputs through autoregressive generation, where each token is drawn from a probability distribution conditioned on prior context. This sampling process introduces inherent stochasticity: Given identical inputs, the LLM model may produce different outputs across invocations. 

This phenomenon motivates a natural question: Can we leverage multiple samples to identify the correct answer more reliably than a single sample, like collective wisdom? The self-consistency approach~\cite{wang2022self} first explored this by sampling multiple reasoning paths from the same model and selecting the most frequent answer, as shown in \autoref{fig:architecture}-(a). The insight is that by aggregating across multiple reasoning samples, the system is more likely to converge to the correct answer. However, it is limited by the knowledge of a single model.

Multi-agent systems extend this principle through two key mechanisms: debate and voting. Diverse agents engage in multi-round discussions where they critique each other's reasoning and progressively refine their proposals~\cite{du2023mad,chai2025scimaster,wang2022self}. Systems~\cite{chen2025tumix,li2025parallelmuse,kaesberg2025voting} like TUMIX~\cite{chen2025tumix} instantiate this pattern by having agents exchange reasoning traces across multiple rounds~\cite{madaan2023self,ji2023reflection}, with the premise that collective deliberation can exceed individual reliability. The decision-making stage usually employs consensus mechanisms like majority voting~\cite{wang2022self,chen2025sets,chen2025tumix,li2025parallelmuse, fu2025deepconf}, where the system commits to whichever answer receives the most votes. Note that majority voting remains the dominant approach due to its interpretability.

\autoref{fig:architecture}-(b) abstracts the common architecture underlying these multi-agent systems.
A coordinator orchestrates a panel of agents, each acting as an independent reasoning node: It receives a problem, generates a candidate solution with supporting rationale, and may revise its answer upon receiving peer solutions. The coordinator aggregates agent outputs after each round and distributes the collective context back to agents for further deliberation. This cycle continues until the termination condition is met.

The combination of parallel exploration and iterative refinement creates systems that significantly outperform single-agent inference on complex reasoning tasks. However, it also creates new problems: When should the system commit to a final result, and what does agreement mean for stochastic agents that may change their positions?

\section{The Multi-Agent Consensus Problem}
\label{subsec:consensus-challenge}

\begin{figure}[!t]
 \centering
 \includegraphics[width=0.45\textwidth]{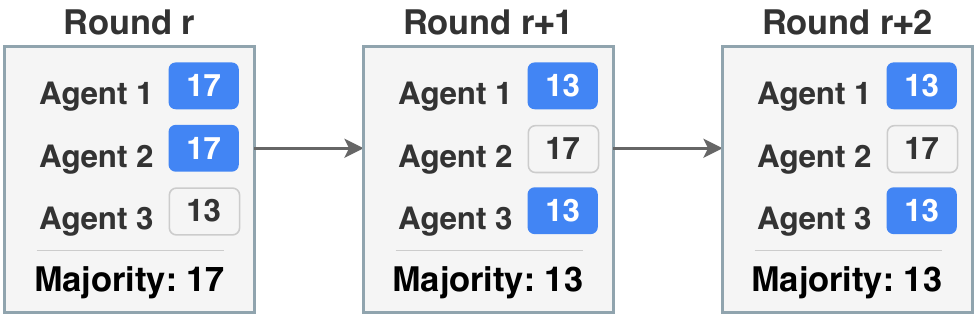}
 \caption{Consensus is unstable; reasoning exchange can flip a majority from $A$ to $B$ between rounds $r$ and $r+1$, rendering termination strategies unreliable.}
 \label{fig:consensus-problem}
\end{figure}

\begin{table}[!t]
\centering
\small
\caption{Convergence efficiency and model latency overhead using GPT-5-mini~\cite{openai_gpt5mini_2025}, Gemini-2.5-Flash~\cite{google_gemini25flash_2025}, and DeepSeek-R1~\cite{deepseekai2025r1}. $L_{\text{fast}}$, $L_{\text{med}}$, and $L_{\text{slow}}$ denote the latency of the fastest, median, and slowest models per round.}
\label{tab:convergence-latency}
\resizebox{0.48\textwidth}{!}{
\begin{tabular}{l|c|ccc|c}
\toprule
& \textbf{Convergence} & \multicolumn{3}{c|}{\textbf{Per-Round Latency (s)}} & \textbf{Barrier} \\
\textbf{Benchmark} & Avg. Rounds & $L_{\text{fast}}$ & $L_{\text{med}}$ & $L_{\text{slow}}$ & Max($L_i$) \\
\midrule
GSM8K & 2.0 & 1.3 & 4.4 & 15.2 & 15.2 \\
MMLU  & 2.3 & 2.5 & 5.8 & 45.0 & 45.0 \\
AIME  & 4.1 & 5.8 & 29.4 & 370.6 & 370.6 \\
\bottomrule
\end{tabular}
}
\end{table}

\noindent
For systems~\cite{chen2025tumix,li2025parallelmuse,kaesberg2025voting} following the architecture in \autoref{fig:architecture}-(b), this termination decision involves two distinct challenges: determining what constitutes effective agreement among stochastic agents, and navigating the trade-off between inference overhead and answer quality. We examine both challenges, demonstrating that classical consensus assumptions are fundamentally incompatible with LLM agents and quantifying the severe performance penalties of current approaches.

\para{The agreement problem: consensus instability.}
Classical consensus protocols~\cite{lamport2001paxos,ongaro2014raft} like Paxos~\cite{lamport2001paxos} guarantee \textit{agreement permanence}: Once a value is accepted by a majority of replicas, the value is permanently committed.
These protocols leverage the \textit{quorum intersection} property to provide such a guarantee despite node failures and network asynchrony.
However, they are agnostic to the semantics of the agreed value;
they permit any proposed value to be committed.

Multi-agent LLM systems face a different problem.
Agents actively revise and refine their solutions based on new context from other agents' reasoning.
An agreement based on simple majority voting does not ensure persistence of the answer, since agents may further refine their solutions.
It also enforces no guarantee on the answer quality --- a majority of agents may incidentally arrive at the same wrong solution.

Consider the example in~\autoref{fig:consensus-problem}.
Three agents deliberate on a mathematical problem.
In round $r$, the agents produce answers $(17, 17, 13)$, yielding a $\frac{2}{3}$ majority for $17$.
However, this majority answer is neither stable nor correct.
After exchanging reasoning traces, Agent 1 recognizes an error in its earlier answer, and the answer distribution shifts to $(13, 17, 13)$ in round $r+1$.
The majority has now shifted from $17$ to $13$.

This instability shows that a simple majority voting approach is unreliable.
\autoref{fig:consensus-problem} illustrates the symptom (a shifting majority), while our experiments (\autoref{tab:agent-exchange}) reveal the root cause: high volatility of agent reasoning.
We observed that agents frequently revise their solutions after reviewing reasoning from peers.
Without proper semantics and protocol designs, a multi-agent system cannot distinguish between coincident agreements that may oscillate and a stable solution that can persist as agents further refine their reasoning.

Issues with current solutions ask for a \emph{formal model} of agreement for non-deterministic agents.
Only with such formal semantics, a multi-agent protocol can reason about the convergence of refinements towards higher quality answers, and determine when a solution has reached stable consensus.

\para{The performance trade-off.}
Even if agreement can be defined, detecting it imposes severe performance costs. The system must balance inference overhead (latency and token consumption) against answer quality. This trade-off manifests in two ways: across rounds (how long to deliberate) and within rounds (how long to wait for agents).

\textit{Across rounds: the inefficiency of fixed limits.} 
Current systems~\cite{wang2025mix} often rely on fixed round limits (e.g., stopping after 4 rounds) to bound computational cost. However, this heuristic ignores the varying complexity of tasks. As shown in the ``Convergence'' column of \autoref{tab:convergence-latency}, both GSM8K and MMLU problems converge rapidly in 2.0 to 2.3 rounds, yet complex reasoning tasks like AIME require 4.1 rounds on average. A fixed limit wastes computation on simple problems while potentially truncating necessary reasoning on difficult ones. Fixed limits thus fundamentally misalign computational cost with the actual difficulty of the problem.

\textit{Within rounds: the bottleneck of barrier synchronization.}
Standard workflows enforce barrier synchronization~\cite{wang2025mix,chen2025tumix}, waiting for all agents to complete a round before evaluating agreement~\cite{wang2025mix,chen2024magicore}. This couples system latency to the slowest straggler. \autoref{tab:convergence-latency} illustrates this disparity: On GSM8K, while the fastest model completes in 1.3s, the system must wait 15.2s for the slowest, resulting in an 11$\times$ slowdown. On MMLU, this gap grows to 18$\times$ (45.0s versus 2.5s). AIME exhibits the most severe disparity, where the slowest model takes 370.6s compared to the fastest at 5.8s, a 64 $\times$ difference. 
Crucially, the straggler identity varies: In our experiments, DeepSeek-R1 was the slowest in 50\% of GSM8K rounds, GPT-5-mini in 42\%, and Gemini in 8\%. This variability makes static exclusion strategies infeasible. When a sufficient subset of agents has already converged, waiting for unpredictable stragglers negates the benefits of parallel exploration without improving the outcome.

Navigating this trade-off requires efficient commitment semantics that finalize an answer as soon as sufficient agreement is detected, without waiting for redundant responses. However, such semantics are only meaningful once the agreement problem is solved. This raises the central question of this paper: How can a multi-agent system determine when agreement is sufficiently stable to commit safely, while still terminating efficiently? Answering this requires a principled framework that formalizes agreement among stochastic agents. No such framework exists today. 

In the following sections, we bridge this gap by formalizing multi-agent refinement as a consensus problem, establishing the theoretical foundation for safe and efficient termination.

\section{Formalizing Multi-Agent Refinement}
\label{sec:model}

\subsection{Problem Formulation}
\label{sec:model:problem}

\noindent
In classical consensus problems, a group of processes proposes and decides on some arbitrary values.
A correct consensus protocol is defined to provide the following guarantees:

\begin{itemize}[noitemsep, topsep=0pt]
    \item \textbf{Validity}: If a correct process decides on a value $v$, $v$ is proposed by a correct process.
    \item \textbf{Agreement}: If values $v$ and $v'$ are decided on any correct processes, then $v = v'$.
    \item \textbf{Termination}: Eventually, every correct process decides some value $v$.
\end{itemize}

Our work targets collaborative multi-agent refinement, which is similar, but not identical, to the consensus problem. 
In our setting, a group of agents is initialized with a common task (how this task is distributed to the agents is irrelevant to the problem).
Each agent can generate a solution to the task based on its local LLM model and context.
When an agent thinks a solution is ready, it can \emph{output} the solution.
Note that an agent can output multiple solutions.
A correct multi-agent refinement protocol shares a similar \emph{termination} property to a consensus protocol:

\begin{itemize}[noitemsep, topsep=0pt]
    \item \textbf{Refinement Termination}: Eventually, some correct agent outputs a solution $s$.
\end{itemize}

The original validity property of the consensus problem is too weak for multi-agent refinement: It accepts a solution even if it is generated by the weakest agent.
Our goal is to generate refined solutions that are better than those generated by individual agents.
Determining which solution is ``better'' requires a function that evaluates the quality of answers.
We define such a function as a \emph{Quality Oracle}:

\begin{definition}[Quality Oracle]
A Quality Oracle $Q$ is a function ($Q : S \times S \rightarrow \mathbb{R}$) that takes two strings as input, where the first string is a task description and the second input is a solution, and outputs a real number representing the quality of the solution to the task.
\end{definition}

We assume $Q$ is a proper function, \ie, given the same input strings, the output value is deterministic.
Note that this quality oracle is unknown to any agent.
To derive a meaningful validity requirement, we define a \emph{majority optimal solution} as follows.
Suppose that each agent can individually generate a solution, a \emph{majority optimal solution} is the best solution (evaluated by $Q$) among the individual solutions generated by a simple majority of agents.
Now, we can define the validity property of multi-agent refinement:

\begin{itemize}[noitemsep, topsep=0pt]
    \item \textbf{Refinement Validity}: 
    The quality of any output solution $s$ is \emph{at least as good} (evaluated by $Q$) as that of \emph{some} majority optimal solution.
\end{itemize}

The usefulness of our refinement validity property depends on the distribution of agent reasoning capabilities.
In a deployment where the majority of agents exhibit strong reasoning capabilities, the property ensures that any output solution has high quality:
By majority intersection, any majority of agents will include at least one strong agent.
Due to the majority optimal property, an output solution will therefore have at least the quality of a strong agent.
Moreover, refinement validity only specifies a \textit{lower bound} on the solution quality.
Practical deployments of multi-agent refinement protocols can produce solutions with stronger properties, as we demonstrate in~\autoref{sec:model:validation}.

The combination of termination and validity creates a problem that cannot be solved trivially.
Without the validity requirement, any agent could simply output its own initial solution.
Without the termination requirement, agents could deliberate indefinitely.
The interaction between these properties requires careful protocol designs to ensure that agents converge to high-quality solutions within bounded time.

Classical consensus requires strict agreement: All correct processes must decide on the same value.
For multi-agent refinement, this requirement is unnecessarily restrictive.
The primary objective is solution quality rather than uniqueness.
We therefore relax agreement to a monotonicity property that permits multiple output solutions, provided that the quality never decreases:

\begin{itemize}[noitemsep, topsep=0pt]
    \item \textbf{Refinement Monotonicity}:
    If a solution $s'$ is outputted after solution $s$ (in real time), the quality of $s'$ is \textit{at least as good} (evaluated by $Q$) as that of $s$.
\end{itemize}

This property ensures that the refinement process makes monotonic progress toward higher-quality solutions.
Even if the protocol generates multiple solutions over time, each subsequent output represents an improvement or maintains the current quality level.

\para{Failure and network models.}
We assume that at most $f$ agents can fail.
We target a \textit{fail-stop model} in which a failed agent stops responding indefinitely but always follows the protocol precisely.
A failure may occur due to server crashes, network partitions, or an indefinite reasoning procedure in an agent.
The network is assumed to be \textit{partially synchronous}~\cite{partialsync}, i.e., there exists a finite but unknown period (commonly known as the global stabilization time) after which all messages are delivered within bounded delay.

\subsection{Agent Operational Semantics}
\label{sec:model:semantics}

\noindent
A key capability of an agent is to perform reasoning (e.g., using an LLM model) based on task descriptions and context to generate solutions.
We now formalize the operational semantics of an agent and the assumptions we make.
Given the stochastic nature of existing AI models, proving formal guarantees of an agent is outside the scope of this work.
Nevertheless, we provide empirical evidence in~\autoref{sec:model:validation} to demonstrate that our semantics and assumptions are sound in practice.

We abstract an agent as a \textit{state machine}, which maintains a context state $c$.
The exact representation of the context state is irrelevant to our problem.
The state machine exposes a reasoning function $R$ with the following definition:

\begin{definition}[Agent Reasoning Function]
    An agent exposes a reasoning function $R$ ($R: c \times S \rightarrow c' \times S$) that takes the current agent context $c$ and a string as input.
    The string is either a task description or a list of reasoning traces.
    $R$ generates a new agent context $c'$, and outputs a string.
    The string includes both a task solution and a reasoning trace of this step.
\end{definition}

In practice, agents implicitly consume their current context and update their context when performing reasoning steps.
Therefore, we omit the context input/output when specifying the reasoning function (e.g., simply writing $s_o \gets R(s_i)$).

When an agent receives a task description, we do not make explicit assumptions about the quality of its reasoning function output (as evaluated by $Q$).
However, we assume that all agents are capable of \emph{reasoning refinement}, \ie,
given a set of reasoning traces and solutions produced by other agents (or possibly from itself), an agent can improve its own reasoning to generate solutions that match the highest quality answer:

\begin{assumption}[Reasoning Refinement]
\label{assump:refine}
    Given a set of solutions and their reasoning traces $S$, suppose an agent applies its reasoning function on a string $\Bar{s}$ that combines all solutions in $S$ to produce a new solution.
    Then the quality of the output solution is \emph{at least as good} (evaluated by $Q$) as any solution in $S$, \ie, $\forall s \in S, Q(R(\Bar{s})) \geq Q(s)$
\end{assumption}

\subsection{Agent Semantics Validation}
\label{sec:model:validation}

\begin{table}[!t]
  \centering
  \footnotesize
  \setlength{\tabcolsep}{4pt} 
  \caption{Accuracy before and after agent exchange. ``Init.'' and ``Post'' denote the accuracy before and after exchanging answers. ``Changes'' indicates the count of decision flips.}
  
  \resizebox{0.45\textwidth}{!}{%
  \begin{tabular}{llcccc}
    \toprule
    \textbf{Dataset} & \textbf{Model} & \textbf{Init.} & \textbf{Post} & \textbf{$\Delta$} & \textbf{Changes} \\
    \midrule
    \multirow{2}{*}{GSM8K}
      & Llama 3.1-8B & 88\% & 96\%    & +8\%  & 17/100 \\
      & DeepSeek-R1          & 98\% & 98\%    & 0\%   &  1/100 \\
    \midrule
    \multirow{2}{*}{MMLU}
      & Llama 3.1-8B & 43\% & 79\%    & +36\% & 64/100 \\
      & DeepSeek-R1          & 62\% & 63\%    & +1\%  &  1/100 \\
    \midrule
    \multirow{2}{*}{AIME}
      & Llama 3.1-8B & 3\%  & 54.4\% & +51\% & 70/90 \\
      & DeepSeek-R1          & 70\% & 72\%    & +2\%  & 10/90 \\
    \bottomrule
  \end{tabular}%
  }
  \label{tab:agent-exchange}
\end{table}

\noindent
Safety and liveness properties of a multi-agent refinement protocol rely on the reasoning refinement assumption.
We now conduct experiments to validate whether real-world AI models conform to this assumption.

We evaluate the reasoning refinement assumption through a collaborative exchange protocol pairing models of different reasoning capacities. A stronger model (DeepSeek-R1) shares its reasoning traces and solutions with a weaker model (Llama3.1-8B-Instruct), and vice versa. Under the refinement assumption, we expect the weaker model to improve by learning from the stronger model's reasoning, while the stronger model should maintain stable performance when exposed to potentially inferior traces. We evaluate on three benchmarks spanning different reasoning datasets: AIME~\cite{aime2024}, MMLU~\cite{hendrycks2021ethics,hendryckstest2021}, and GSM8K~\cite{cobbe2021gsm8k}.

\autoref{tab:agent-exchange} summarizes the results. The weaker model exhibits substantial accuracy improvements across all benchmarks after receiving reasoning traces from DeepSeek-R1: Accuracy increases from 3\% to 54.4\% on AIME-Validation (+51\%), from 43\% to 79\% on MMLU (+36\%), and from 88\% to 96\% on GSM8K (+8\%). In contrast, the stronger model maintains remarkable stability, with minimal accuracy changes on all datasets. The low decision change counts for DeepSeek-R1 (1/100 on GSM8K, 1/100 on MMLU, 10/90 on AIME) indicate that strong models can filter inferior reasoning rather than blindly adopting it.

These results validate two properties central to \autoref{assump:refine}. First, weaker models benefit from superior reasoning traces, achieving significant improvement. Second, strong models maintain quality when observing inferior reasoning, confirming quality-preserving refinement. This asymmetry is crucial for protocol correctness: Agents selectively adopt reasoning only when it represents an improvement, ensuring that collective deliberation produces monotonically improving reasoning. These empirical findings justify our reliance on the reasoning refinement assumption for the protocol guarantee.

\section{\sys Protocol}
\label{sec:protocol}

\noindent
In this section, we introduce \sys, a concrete protocol that correctly implements multi-agent refinement (\autoref{sec:model}).

\subsection{Protocol Overview}
\label{sec:overview}

\noindent
\para{Agent ensemble.} An \sys deployment consists of $N$ agent processes, denoted $A_i$ where $i$ ranges from $0$ to $N-1$, forming an LLM ensemble. Agents have complete network connectivity.
We assume a partially synchronous~\cite{partialsync} network, and at most $\lceil \frac{N-1}{2} \rceil$ agents may suffer fail-stop failures.
An \sys deployment takes as input a single task description, and outputs one or multiple solutions to the task.
The protocol guarantees all properties of multi-agent refinement, as defined in \autoref{sec:model}.
The protocol can be easily extended to a multi-instance version;
details of such a protocol variant are omitted.

The protocol progresses in \emph{terms}.
In each term, there is at most one leader agent.
We use standard leader election protocols~\cite{lamport2001paxos,pmmc,ongaro2014raft} to elect a leader agent.
When the leader agent in a term becomes unreachable, the remaining agents advance to the next term and start a new leader election.
If a leader election fails to make progress, the agents move to the next term and retry a new election.
The first term begins with a predetermined leader or a leader election protocol.

In each term, the leader coordinates all agents to generate refined solutions.
The protocol within a term is divided into \emph{rounds}.
All rounds share the same message pattern: The leader distributes the quorum refinement set from the previous round to all agents; each agent performs refinement and sends the refined solution to the leader; the leader aggregates a quorum of refined solutions for the next round.
Leveraging the reasoning refinement (\autoref{assump:refine}) and the quorum intersection properties, the quality of solutions improves monotonically in each round.
Once the leader receives a quorum of solutions in a round, it can output a solution from the refinement set in the previous round.
Such outputs guarantee refinement validity and monotonicity, even in the event of leader failures.
To decide when and which solution to output from refinement sets, the protocol incorporates a modular \emph{refinement decision engine}.
Users can customize the solution selection criteria in the engine.
\autoref{fig:execution-timeline} and \autoref{alg:multiagent-consensus} show the protocol's refinement timeline within each term and the overall workflow.

When a new leader is elected in a new term, the leader collects a quorum of refinement sets before proceeding to the first round of the term.
This ensures that the quality of solutions in the new term exceeds any previous output.
The design resembles existing leader and view change protocols~\cite{oki1988viewstamped,pmmc}.

\subsection{Protocol Specification}
\label{sec:protocol-spec}

\begin{figure}[!t]
 \centering
 \vspace{1pt}
 \includegraphics[width=0.47\textwidth]{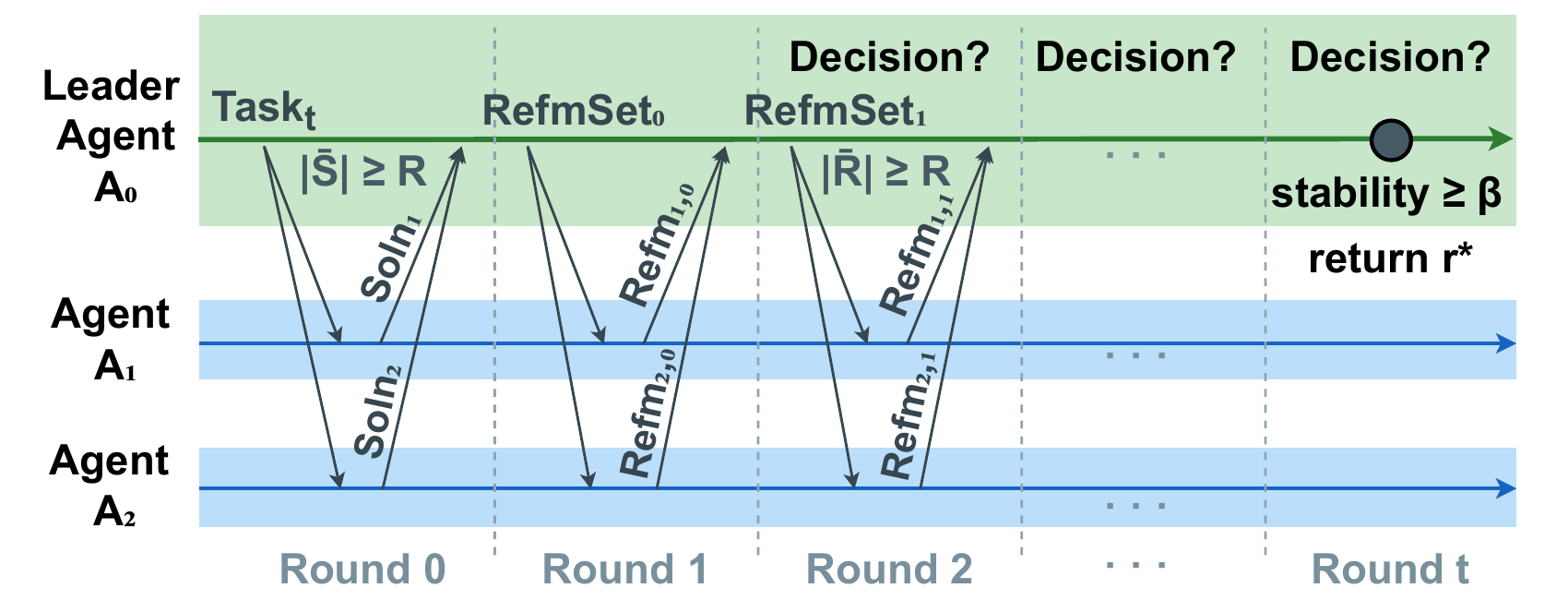}
\caption{Execution timeline of the \sys protocol within one term. Initially, the leader broadcasts the task and gathers solutions until quorum ($|\bar{S}| \geq R$). During refinement, the leader broadcasts the reference set and workers respond with their references. This process repeats until a candidate is established ($|\bar{R}| \geq R$). Finalization occurs when the stability counter reaches $\beta$. For brevity, we abbreviate $\langle$Task, \textit{term-num}, $t\rangle$ as Task$_t$, $\langle$Soln, \textit{term-num}, \textit{id}, $s\rangle$ as Soln$_{id}$, $\langle$RefmSet, \textit{term-num}, \textit{round-num}, $\bar{R}\rangle$ as RefmSet$_{num}$, and $\langle$Refm, \textit{term-num}, \textit{id}, \textit{round-num}, $r\rangle$ as Refm$_{id,num}$.}
 \label{fig:execution-timeline}
\end{figure}

\begin{algorithm}[!h]
\caption{Multi-Agent Refinement Protocol}
\label{alg:multiagent-consensus}
\small
\KwIn{Task $t$; agent pool $\mathcal{A}$}
\KwOut{Refined solution $r^*$}
\BlankLine
$r^* \gets \bot$; $\mathcal{L} \gets \bot$; $\Bar{R} \gets \bot$; round-num $\gets 0$\;\label{line:init}
\If{$\mathcal{L} = \bot$ \textbf{or} $\mathcal{L}$ unreachable}{
    $\mathcal{L} \gets$ LeaderElection()\;\label{line:election}
}
$\mathcal{L}$ broadcasts \msg{Task}{term-num, t}\;\label{line:task-dist}
$\forall a \in \mathcal{A}: s \gets \mathcal{R}_a(t)$; Sends \msg{Soln}{term-num, id, s} to $\mathcal{L}$\;\label{line:sol}
$\mathcal{L}$ collects quorum \mtype{Soln}; $\Bar{R} \gets$ set of $s$ in the \mtype{Soln} quorum\;\label{line:sol-collect}
\While{$r^* = \bot$}{
  round-num $\mathrel{+}= 1$\;\label{line:advance}
  $\mathcal{L}$ broadcasts \msg{RefmSet}{term-num, round-num, $\Bar{R}$}\;\label{line:broadcast}
  $\forall a \in \mathcal{A}: r \gets \mathcal{R}_a(\Bar{R})$; Sends \msg{Refm}{term-num, id, round-num, r} to $\mathcal{L}$\;\label{line:refn}
  $\mathcal{L}$ collects quorum \mtype{Refm}; $\Bar{R} \gets$ set of r in the \mtype{Refm} quorum\;\label{line:collect}
  $\mathcal{L}$ inputs ($\Bar{R}$, round-num) into decision engine\;\label{line:decision}
}
\textbf{return} $r^*$\;\label{line:return}
\end{algorithm}

\begin{table}[!h]
\centering
\small
\caption{Protocol state variables maintained by each agent.}
\begin{tabular}{@{}ll@{}}
\toprule
\textbf{Variable} & \textbf{Description} \\
\midrule
\state{term-num} & Current term number \\
\state{role} & Role: \state{leader}, \state{candidate}, or \state{worker} \\
\state{id} & Unique agent ID \\
\state{round-num} & Current refinement round within term \\
\state{refmset} & Most recent refinement set $\bar{R}$ received \\
\bottomrule
\end{tabular}
\label{tab:agent-state}
\end{table}

\para{Agent protocol state.}
We show the protocol state each agent maintains in \autoref{tab:agent-state}.
In each term, at most one agent can have \state{role} set to \state{leader}.
Agents tag their current \state{term-num} in all protocol messages.

\para{Leader election.}
\sys uses a leader election protocol similar to Raft~\cite{ongaro2014raft} (\autoref{line:election}).
Unlike Raft, any agent is eligible for the leader role in a term;
the leader change protocol ensures protocol safety.
When an agent enters a new term, it broadcasts a \msg{RequestVote}{term-num} message to all agents and transitions to \state{candidate} state.
When an agent receives a \mtype{RequestVote} with a higher \state{term-num}, it advances its \state{term-num}, transitions to \state{worker}, and replies a \msg{Vote}{term-num, id} to the candidate.
Once a candidate receives a quorum of \mtype{Vote} from unique agents (it always votes for itself) for its current term, it transitions to \state{leader} state.
An agent can only cast one \mtype{Vote} in each term.
If the protocol fails to make progress, agents will time out, advance to the next term, and retry leader election.

\para{Refinement protocol.} As shown in \autoref{fig:execution-timeline}, all agents receive an identical task $t$ before the protocol starts.
Details of this task distribution are irrelevant to \sys.
When an agent transitions to the \state{leader} role in the first term, it broadcasts a \msg{Task}{term-num, t} to all agents (\autoref{line:task-dist}).
Upon receiving the message, a worker agent applies its reasoning function $\mathcal{R}$ on task $t$ to derive its initial solution $s$.
Then, it sends a \msg{Solution}{term-num, id, s} to the leader (\autoref{line:sol}).
The leader collects a quorum of \mtype{Solution} from unique agents (including from itself), and broadcasts to all agents a \msg{RefmSet}{term-num, rounds-num, $\Bar{R}$}, where \state{round-num} is set to $1$ and $\Bar{R}$ contains all the solutions in the quorum (\autoref{line:sol-collect}).

Once an agent receives a \mtype{RefmSet} with a \state{round-num} no less than its own, it stores $\Bar{R}$ in its \state{refmset} and updates its \state{round-num}.
It then applies its reasoning function $\mathcal{R}$ on $\Bar{R}$ to derive a refined solution $r$.
Next, the agent sends a \msg{Refm}{term-num, id, round-num, r} to the leader (\autoref{line:refn}).
The leader collects a quorum of \mtype{Refm} from unique agents with \state{round-num} matching its own, and inputs the set of refinements $\Bar{R}$ and the \state{round-num} to the refinement decision engine (\autoref{line:decision}).
Then, the leader advances the \state{round-num} (\autoref{line:advance}) and again broadcasts a \msg{RefmSet}{term-num, round-num, $\Bar{R}$} to all agents. 

The workers and the leader repeat the refinement loop until the refinement decision engine outputs a solution (\autoref{line:return}).
The leader then returns the solution to the client.

\para{Leader change.}
When a new leader is elected, it broadcasts a \msg{NewTerm}{term-num} to all agents.
If an agent receives a \mtype{NewTerm} with a \state{term-num} greater than its own, it replies a \msg{NewTermAck}{term-num', term-num, id, round-num, $\Bar{R}$} to the new leader, where \state{term-num'} is the new term and $\Bar{R}$ is its current \state{refmset}.
It then advances its \state{term-num} to that in \mtype{NewTerm}.
When the new leader receives a quorum of \mtype{NewTermAck} with \state{term-num'} matching its own, it finds the \mtype{NewTermAck} with the highest \state{term-num}, and in case of ties, the one with the highest \state{round-num}.
Next, the leader sets its \state{round-num} to $1$, and broadcasts a \msg{RefmSet}{term-num, round-num, $\Bar{R}$} to all agents, where $\Bar{R}$ is from the selected \mtype{NewTermAck}.
The agents and the leader then proceed with the normal refinement protocol.

\para{Refinement decision engine.}
The protocol defines a refinement decision engine for determining \emph{when} to output a solution and \emph{which} solution to output.
The engine takes as input a \state{term-num}, a \state{round-num}, and a set of refined solutions $\Bar{R}$.
Any solution in a refinement set satisfies refinement validity.
However, to guarantee refinement monotonicity, the engine can only output a solution in round $i$ once it receives a refinement set in round $i+1$.

Beyond the minimum protocol safety requirements, users can define additional criteria for outputting a solution.
For instance, a user can enforce a similarity threshold $\alpha$:
The engine only emits a solution if at least $\alpha$ solutions in a refinement set are semantically equivalent to the solution.
Semantic equivalence is evaluated using standard vector embedding similarity or exact match. Requiring multiple agents to agree filters stochastic errors that affect individual outputs.
A user can also apply a stability horizon threshold $\beta$:
The engine outputs a solution only after the solution or its semantical equivalence appears in $\beta$ consecutive rounds of the refinement sets. Equivalence across rounds is determined using exact match, LLM-as-judge, or both based on tasks. $\beta$-round persistence ensures the agreement survives continued deliberation rather than transient coincidence.

\begin{figure}[!t]
 \centering
 \includegraphics[width=0.48\textwidth]{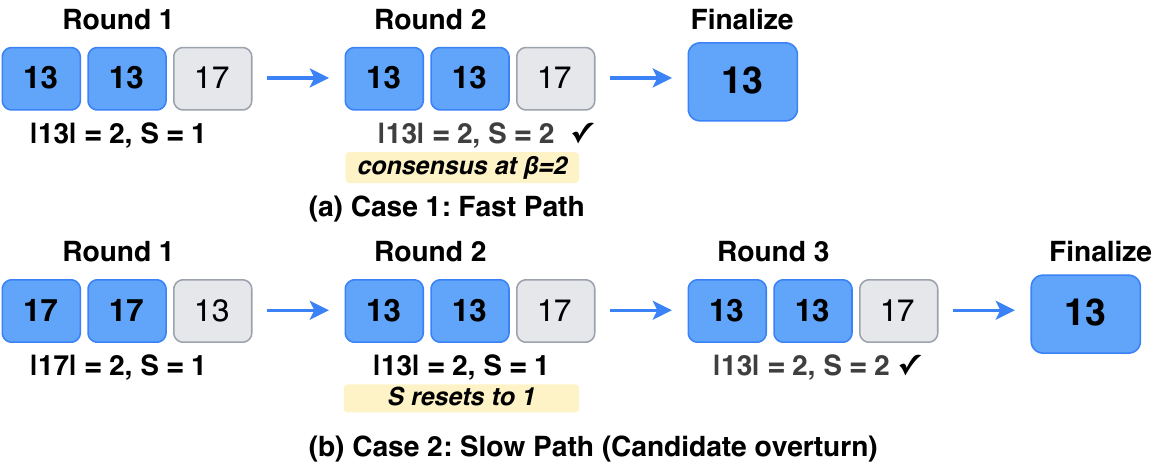}
 \caption{\sys consensus execution examples with $\alpha = 2$ and $\beta = 2$ among 3 agents. Three boxes in each round represent responses from three agents.}
 \label{fig:execution-example}
\end{figure}

\autoref{fig:execution-example} illustrates two consensus scenarios with $\alpha, \beta$ termination condition. Case 1 shows a fast consensus path where two agents initially agree on value 13 in Round 1. This majority remains stable through Round 2, satisfying the $\beta = 2$ stability requirement for finalization. Case 2 demonstrates a slow consensus path where stability validation proves necessary. Two agents initially agree on value 17 in Round 1, but the majority shifts to value 13 in Round 2. This overturn resets the stability counter to 1. By Round 3, value 13 maintains majority support for two consecutive rounds, finally meeting the $\beta = 2$ requirement and allowing finalization.

\subsection{Correctness Proof}
\label{sec:correctness}

\noindent
We now prove that the \sys protocol guarantees multi-agent refinement safety (refinement validity and monotonicity) and liveness (refinement termination) as defined in \autoref{sec:model:problem}.

\begin{lemma}[Refinement Monotonicity]
\label{lem:monotonicity}
If \sys outputs a solution $s'$ after solution $s$ in real time, $s'$ has equivalent or better quality than $s$.
\end{lemma}

\begin{proof}
First, we show that \sys outputs at most one solution in a round within a term.
Each term has at most one leader due to the leader election protocol.
Since only a leader can output a solution, and a leader outputs at most one solution in a round, the proposition is true.

Next, we show that within a term, any output solution $s'$ in round $i+1$ has equivalent or better quality than any solution $s$ in round $i$.
We denote the refinement set in round $i$ by $\Bar{R_i}$.
An output solution $s$ in round $i$, if exists, is contained within $\Bar{R_i}$.
In round $i+1$, any refined solution $r$ is generated by performing the agent reasoning function on $\Bar{R_i}$.
Given the reasoning refinement assumption (\autoref{assump:refine}), the quality of $r$ is at least as good as $s$.
As such, any solution in the round $i+1$ refinement set $\Bar{R}_{i+1}$ has quality equal to or better than $s$.
Since any output solution $s'$ in round $i+1$ is contained within $\Bar{R}_{i+1}$, the quality of $s'$ is at least as good as $s$.
Because output solutions in real time order obey round order, the lemma holds within a round.

Lastly, we show that any output solution $s'$ in term $t+1$ has equivalent or better quality than any output solution $s$ in term $t$.
Given the proposition we just proved, we only need to show the case for the last output solution $s$ in term $t$ and the first output solution $s'$ in term $t+1$.
Since $s$ is an output solution, $s$ is contained within a refinement set $\Bar{R}$, and a quorum of agents has stored $\Bar{R}$ locally.
When the leader in term $t+1$ receives a quorum of \mtype{NewTermAck}, at least one \mtype{NewTermAck} contains $\Bar{R}$ due to quorum intersection.
As $s$ is the last output solution in term $t$, the leader will include $\Bar{R}$ in its initial \mtype{RefmSet}.
Following the same reasoning in the proof of the previous proposition, any first output solution $s'$ in term $t+1$ has equivalent or better quality than any solution in $\Bar{R}$, including $s$.
\end{proof}

\begin{lemma}[Refinement Validity]
\label{lem:validity}
If \sys outputs a solution $s$, $s$ satisfies refinement validity.
\end{lemma}

\begin{proof}
Prove by induction.
We first prove the base case.
The first leader collects a quorum of solutions $\Bar{S}$ from distinct agents.
Suppose the highest quality solution in $\Bar{S}$ is $s^*$, $s^*$ is then a majority optimal solution.
In the first round, any refined solution $r$ is generated by performing the agent reasoning function on $\Bar{S}$.
Given the reasoning refinement assumption (\autoref{assump:refine}), the quality of $r$ is at least as good as any solution in $\Bar{S}$, including $s^*$.
Since any output solution $s$ in the round is selected from these refined solutions, $s$ is at least as good as $s^*$.
This proves the base case.
The inductive step is trivially true following \autoref{lem:monotonicity}.
\end{proof}

\begin{lemma}[Refinement Termination]
\label{lem:termination}
Eventually, \sys outputs a solution $s$.
\end{lemma}

\begin{proof}
\sys assumes a partially synchronous network model.
After periods of asynchrony, all messages are delivered within bounded time.
The leader election protocol is guaranteed to elect a leader.
Given that at most $f$ agents can fail, eventually a correct agent becomes the leader.
Since at least a quorum of agents are alive, this leader can collect a quorum of \mtype{RefmSet} within a round.
After two refinement rounds, the leader can output a solution.
\end{proof}

\begin{theorem}[Correctness]
\label{theorem:correctness}
The \sys protocol correctly implements multi-agent refinement.
\end{theorem}

\begin{proof}
The correctness theorem follows directly from \autoref{lem:monotonicity}, \autoref{lem:validity}, and \autoref{lem:termination}.
\end{proof}

\section{Consensus-Aware Serving Infrastructure}
\label{sec:runtime}

\noindent
The protocol specification in \autoref{sec:protocol} builds formal consensus but abstracts away execution concerns. This section presents \sys-Serve, a serving infrastructure that exposes consensus semantics to the scheduling layer, enabling optimizations that neither consensus-oblivious serving systems nor serving-oblivious consensus implementations can achieve.

\subsection{System Architecture}
\label{subsec:arch}

\begin{figure}[!t]
    \centering
    \includegraphics[width=0.45\textwidth]{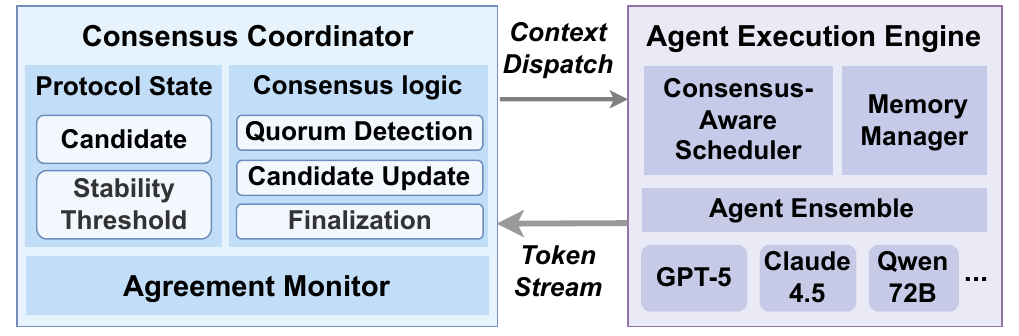}
    \caption{The architecture of \sys-Serve. The consensus coordinator maintains protocol state and detects quorum. It directs the agent execution engine, which optimizes resource usage through ensemble-aware scheduling.} 
    \label{fig:serving-arch}
\end{figure}

\noindent
As shown in \autoref{fig:serving-arch}, \sys employs a two-layer architecture. The \textit{consensus coordinator} maintains protocol state including the current candidate answer, stability threshold, and round history. It executes the consensus logic from \autoref{alg:multiagent-consensus}, tracking when a value achieves quorum and monitoring progress toward the stability horizon. Its embedded \textit{agreement monitor} tracks agent responses incrementally and evaluates quorum predicates as answers arrive, enabling early termination without barrier synchronization. The \textit{agent execution engine} manages inference through a \textit{consensus-aware scheduler} that treats agent ensembles as first-class scheduling units, and a \textit{memory manager} that handles KV cache allocation across the heterogeneous agent pool.

\para{Universal engine abstraction.}
To enable \sys's optimizations across diverse LLM backends with varying model architectures, tokenizers, and memory layouts, we define a minimal interface that execution engines must implement:

\begin{lstlisting}[style=PythonStyle]
def Dispatch(q: str, ctx: ConsensusCtx, eid: int) -> Handle
def OnComplete(h: Handle, cb: Callable[[Answer], None])
def Cancel(h: Handle) -> bool
def QueryEnsemble(eid: int) -> EnsembleState
\end{lstlisting}

These methods expose both inference functionality and consensus-aware coordination. The \texttt{Dispatch} method initiates agent execution with attached consensus context (current candidate, stability threshold, round number), returning a handle for subsequent operations. The \texttt{OnComplete} method registers a callback invoked upon generation completion, delivering the normalized answer to the agreement monitor. The \texttt{Cancel} method terminates in-progress generation for early quorum exploitation, returning resources to the memory pool. Finally, \texttt{QueryEnsemble} retrieves aggregate state for all agents sharing an ensemble identifier, enabling collective admission decisions and ensemble-aware preemption policies. This abstraction decouples consensus logic from backend-specific implementation details, allowing \sys to integrate with systems like vLLM~\cite{vllm2025} with minimal adaptation.

The layers interact through an event-driven design. The coordinator dispatches execution requests annotated with consensus context. As agents complete generation, the engine delivers normalized answers to the agreement monitor. When the monitor detects that at least $\alpha$ agents agree on the same answer and meet the termination criteria, it signals the coordinator to update the candidate. It also issues cancellation directives for in-progress agents whose answers can no longer affect the round outcome.

\subsection{Ensemble-Aware Scheduling}
\label{subsec:ensemble-scheduling}

\noindent
Traditional LLM serving systems treat each inference request as an independent scheduling unit. This abstraction is insufficient for consensus-based multi-agent execution, where the outputs of multiple agents are collectively evaluated. We introduce the \textit{agent ensemble} as a first-class scheduling primitive: a set of $n$ inference requests that share a common query, whose outputs will be aggregated for quorum evaluation.

\para{Collective admission and scheduling.} 
Rather than admitting requests individually, when a new query arrives, the scheduler evaluates whether sufficient resources exist for the entire ensemble before admitting any member. This collective policy prevents scenarios where partially admitted ensembles consume resources while waiting indefinitely for peers that cannot be scheduled. Specifically, the scheduler checks that GPU memory can accommodate $n$ concurrent KV caches and that system load allows at least $\alpha$ agents to complete within the round timeout.
Once admitted, the scheduler co-schedules the ensemble to synchronize their progress through the decode phase. This minimizes completion time variance, reducing the window during which completed agents hold resources while awaiting others and enabling earlier quorum detection.

\para{Ensemble-aware termination and preemption.}
The agreement monitor implements progressive cancellation to exploit early consensus within an ensemble. Upon each agent completion, the monitor increments the support count for the returned answer. When a solution receives $\alpha$ votes and sustains for $\beta$ rounds, the coordinator declares round termination and cancels in-progress agents through the scheduler's abort interface, returning their KV cache blocks to the free pool immediately.

Also, under memory pressure, request-level preemption policies can fragment ensembles by evicting some agents while others continue, potentially leaving insufficient agents for a quorum. Our ensemble-aware policy treats the ensemble atomically: Before preempting any agent, the scheduler verifies that at least $\alpha$ agents would survive. If this condition fails, the scheduler either preempts the entire ensemble for later retry or protects near-completion ensembles from preemption.

\label{subsec:protocol-aware}
\para{Protocol-aware resource management.}
The protocol's consensus state enables predictive resource allocation. Once a candidate is established with $s$ consecutive rounds of stability, at most $S-s$ additional rounds remain before finalization. The coordinator exploits this bounded horizon by sending \textit{reservation hints} specifying expected remaining rounds and reduced ensemble size ($R+1$ agents rather than $n$), preventing nearly-finalized queries from being starved by new arrivals. Conversely, upon finalization, an \textit{eager release directive} triggers immediate resource reclamation rather than waiting for garbage collection.

\subsection{Failure Handling}
\label{subsec:failures}

\noindent
The coordinator detects hard failures (crashed processes) through periodic heartbeats and soft failures (stalled generation) through completion timeouts. Upon failure, if the remaining healthy agents $H_t$ satisfy $|H_t| \geq \alpha$, execution continues normally with the failed agent marked as non-participating. When $|H_t| < \alpha$, the coordinator aborts and restarts if no candidate exists; otherwise, it preserves the current candidate and transitions to the next round with a fresh ensemble. This policy ensures failures never violate safety: A candidate value remains a candidate, and finalization occurs only after the required $S$ consecutive confirmations, regardless of failures. 

\section{Experiment}
\label{sec:eval}

\noindent
We implement \sys atop vLLM~v0.10~\cite{vllm2025}, modifying around  2000 lines of Python code. 

\subsection{Experimental Setup} 
\label{sec:exp-setup}

\para{Testbed.}
To evaluate \sys under realistic serving conditions, we deploy the system on 8$\times$ NVIDIA H100 GPUs (80GB) using vLLM v0.10~\cite{vllm2025} as the inference backend for locally hosted models. For commercial model evaluation, we access them through the OpenRouter API~\cite{openrouter2025}, which provides a unified interface to multiple providers while preserving their native latency characteristics.

\para{Agent ensemble.} 
We build an agent pool to reflect production deployments where models of varying capabilities collaborate. The pool comprises five open-source models deployed locally: Llama-3.1-8B-Instruct~\cite{llama3-8b}, Qwen3-8B~\cite{qwen3-8b}, GPT-OSS-20B~\cite{gptoss20b}, Mistral-7B-Instruct~\cite{mistral7b}, DeepSeek-Qwen-2.5-7B~\cite{deepseek_qwen7b}, each allocated to a dedicated H100. We also evaluate three commercial models of comparable capability (GPT-5-mini~\cite{openai_gpt5mini_2025}, Gemini-2.5-flash~\cite{google_gemini25flash_2025}, Claude-4.5-Haiku~\cite{anthropic_claudehaiku45_2025}) via public API to demonstrate \sys's generalizability beyond local deployments.

\para{Benchmarks and workload.}
We select four widely-adopted reasoning benchmarks for multi-agent systems of varying difficulty: GSM8K~\cite{cobbe2021gsm8k}, 
MMLU~\cite{hendrycks2021ethics,hendryckstest2021}, 
AIME~\cite{aime2024} 
and IMO datasets~\cite{luong-etal-2025-towards}. This diversity allows us to examine whether \sys adapts computational effort to problem complexity. To simulate production traffic patterns, query submitters generate Poisson arrivals with varying request rates.

\para{Baselines and configuration.}
We compare \sys against the following baselines. \textit{Multi-Agent-Base} faithfully reproduces the core mechanisms of state-of-the-art multi-agent systems~\cite{chen2025tumix,li2025parallelmuse,chen2025sets,du2023mad,kaesberg2025voting} like TUMIX~\cite{chen2025tumix}. Multi-Agent-Base adopts per-round barrier synchronization where all agents must complete before proceeding, heuristic-based round limits, and majority voting for final answer selection. \textit{Best/Worst single model} selects the highest/lowest-accuracy model within the LLM ensemble of each experiment, establishing single-model performance bounds.
We implement all baselines in a unified codebase with the same agent ensemble; \sys differs only in consensus detection and termination strategy. Unless otherwise specified, \sys uses default parameters: $n{=}3$ agents, $\alpha{=}2$ (majority quorum size), $\beta{=}2$ rounds (stability horizon threshold), $T_{\max}{=}5$ rounds.

\para{Metrics.}
We report metrics that capture both efficiency and accuracy. Average latency and P99 tail latency characterize typical and worst-case response times, respectively. Accuracy measures the fraction of problems answered correctly.

\begin{figure*}[!t]
    \centering
    \includegraphics[width=0.9\textwidth]{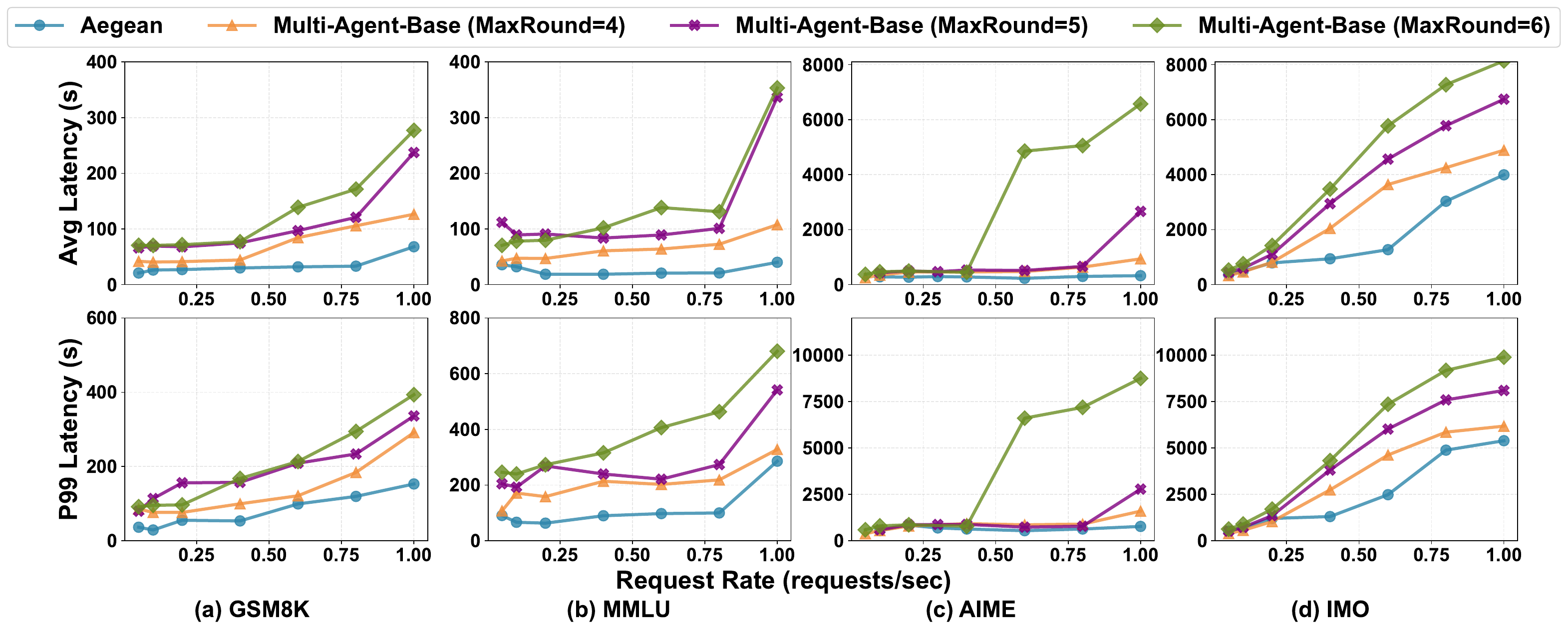}
    \caption{Throughput versus latency on locally deployed models across four benchmarks, comparing \sys to Multi-Agent-Base with varying maximum rounds.}
    \label{fig:perf-local}
\end{figure*}

\subsection{End-to-End Performance}
\label{sec:end-to-end}

\noindent
Our evaluation first verifies that \sys achieves its primary design goal of reducing latency while preserving accuracy compared to baseline systems.

\para{Local model deployment.}
\autoref{fig:perf-local} compares \sys against Multi-Agent-Base with varying maximum rounds (4, 5, 6) on four reasoning benchmarks using locally deployed models. The top row presents average latency while the bottom row shows P99 latency, capturing both typical and tail-case performance characteristics.

For average latency, \sys maintains consistently low values even as request rate increases from 0.1 to 1.0 requests/sec. On simpler benchmarks like GSM8K and MMLU, \sys sustains latency below 100s throughout all load levels, while Multi-Agent-Base degrades substantially at 1.0 requests/sec: On GSM8K, MaxRound=4/5/6 configurations reach 126s, 237s, and 277s respectively (1.9$\times$, 3.5$\times$, 4.1$\times$ higher than \sys); on MMLU, they reach 108s, 337s, and 353s (2.7$\times$, 8.4$\times$, 8.8$\times$ higher). For more challenging benchmarks, \sys achieves latency around 325s on AIME and 3993s on IMO at the highest request rate. In contrast, Multi-Agent-Base exhibits severe latency degradation under load: On AIME at 1.0 requests/sec, MaxRound=4/5/6 reach 937s, 2662s, and 6571s (2.9$\times$, 8.2$\times$, and 20.2$\times$ higher); on IMO, they reach 4889s, 6743s, and 8138s (1.2$\times$, 1.7$\times$, 2.0$\times$ higher).

The P99 latency results further highlight the advantage of quorum-based execution. On GSM8K and MMLU, \sys maintains P99 latency of 153s and 286s at maximum load. In comparison, Multi-Agent-Base shows consistently higher tail latency: On GSM8K, MaxRound=4/5/6 reach 291s, 336s, and 394s (1.9$\times$, 2.2$\times$, and 2.6$\times$ higher); on MMLU, they reach 328s, 542s, and 681s (1.1$\times$, 1.9$\times$, and 2.4$\times$ higher). The gap widens dramatically on harder benchmarks: At 1.0 requests/sec, \sys achieves P99 latency of 772s on AIME, while MaxRound=4/5/6 reach 1595s, 2782s, and 8749s (2.1$\times$, 3.6$\times$, and 11.3$\times$ higher). On IMO, \sys achieves 5391s compared to 6175s, 8097s, and 9888s for MaxRound=4/5/6 (1.1$\times$, 1.5$\times$, and 1.8$\times$ higher). The controlled tail latency shows that \sys provides bounded worst-case performance, critical for meeting service level objectives in production.

\autoref{fig:perf-local} reveals distinct saturation behavior between the two approaches. Multi-Agent-Base shows exponential latency growth as the system queues requests waiting for slow agents to complete all rounds. In contrast, \sys's quorum-based termination processes queries without waiting for stragglers, resulting in sub-linear growth under increasing load. The performance gap becomes more pronounced on harder benchmarks like AIME and IMO, where individual agent response times exhibit higher variance.

\begin{figure}[!t]
    \centering
    \includegraphics[width=0.45\textwidth]{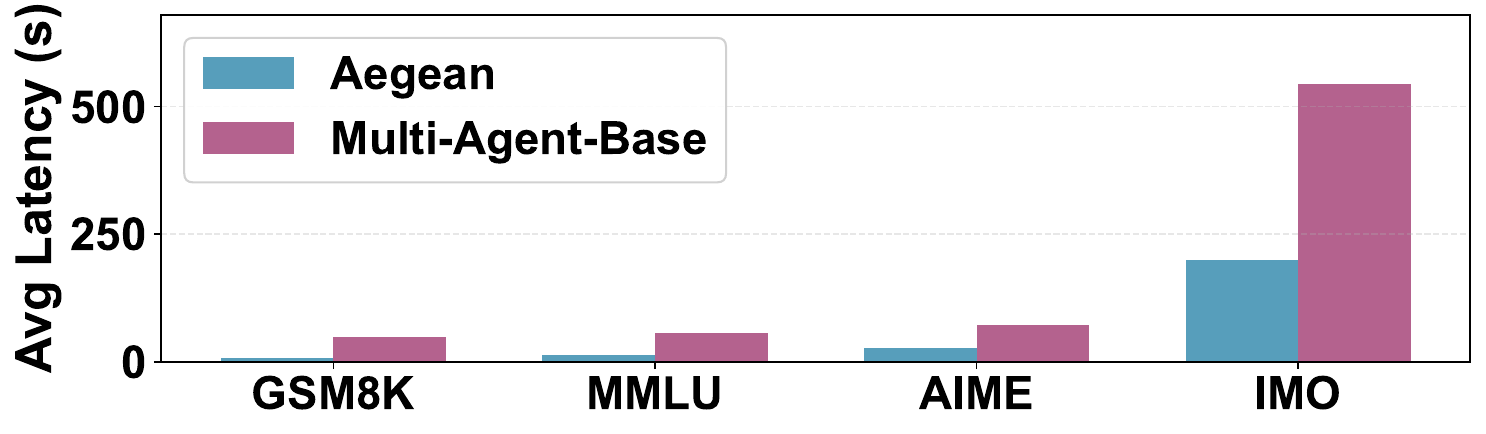}
    \caption{Average latency comparison using commercial frontier models accessed via public API.}
    \label{fig:perf-api}
\end{figure}

\para{Public API deployment.} Given that production systems increasingly rely on commercial APIs, we further evaluate on cloud provider APIs.
\autoref{fig:perf-api} evaluates the same systems using commercial frontier models (GPT-5-mini, Claude-4.5-haiku, Gemini-2.5-flash) accessed via OpenRouter API. \sys achieves significant latency reduction across all benchmarks: 8.0s vs 49.2s on GSM8K (6.2$\times$ speedup), 13.8s vs 57.2s on MMLU (4.1$\times$ speedup), 26.4s vs 72.5s on AIME (2.7$\times$ speedup) and 198.8s vs 543.6s on IMO (2.7$\times$ speedup). We use a fixed request rate of 0.5 requests/sec for API experiments because commercial providers maintain large GPU clusters that absorb load fluctuations, resulting in stable per-request latency. The latency improvements instead stem from two inherent characteristics of commercial APIs. First, commercial APIs exhibit higher variance in response times due to shared infrastructure and rate limiting, making straggler delays more pronounced in wait-all systems. Second, frontier models generate longer responses on average, amplifying the cost of waiting for all agents to complete. \sys's consensus protocol effectively masks this variance by committing as soon as a quorum converges, providing predictable latency regardless of underlying API fluctuations. The consistent latency reduction across both deployment modes confirms that quorum-fast execution addresses the performance bottleneck in multi-agent systems.

\begin{figure}[!t]
    \centering
    \includegraphics[width=0.48\textwidth]{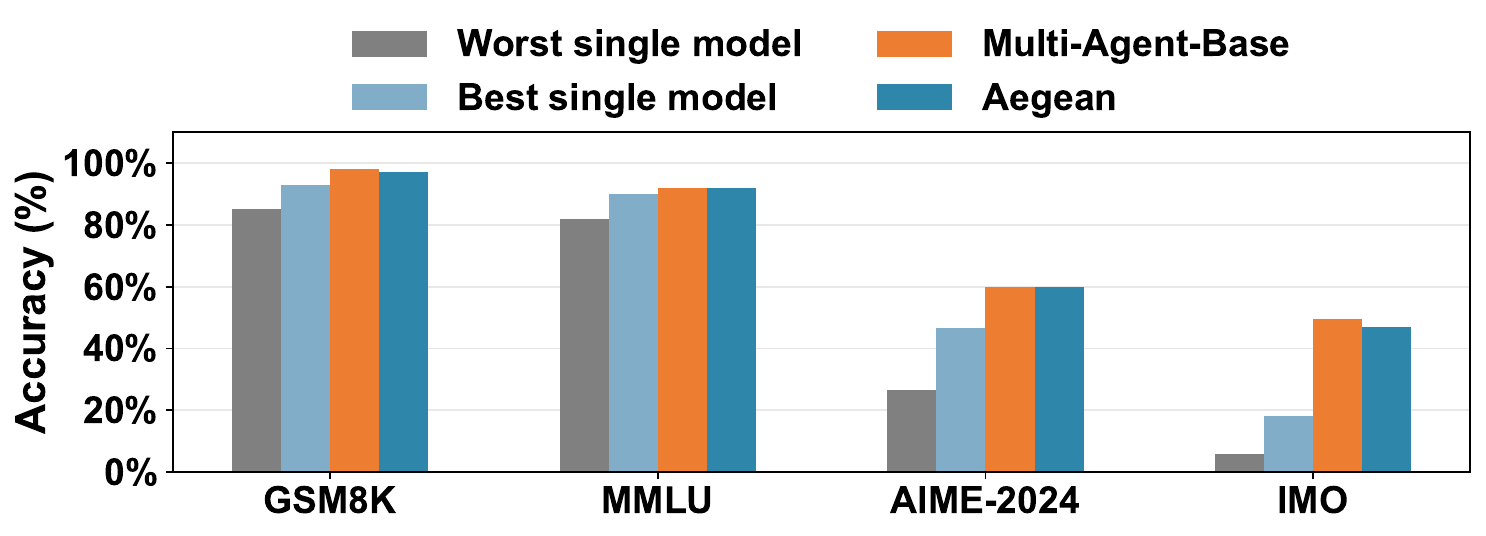}
    \caption{Accuracy comparison among best single model, Multi-Agent-Base and \sys.}
    \label{fig:acc-overall}
\end{figure}

\para{Accuracy results.} 
Critically, \sys maintains accuracy comparable to Multi-Agent-Base while achieving these latency improvements. 
\autoref{fig:acc-overall} compares accuracy across four benchmarks for worst/best single models, Multi-Agent-Base, and \sys. 
On GSM8K, both multi-agent approaches significantly outperform single models, with Multi-Agent-Base achieving 98\% and \sys achieving 97\% compared to 93\% for the best single model. 
MMLU shows similar patterns. On more challenging benchmarks, the benefits of multi-agent collaboration become more pronounced. For AIME, both multi-agent approaches achieve 60\% accuracy compared to 46.7\% for the best single model, representing a 28\% relative improvement. IMO exhibits the largest gap: Multi-Agent-Base reaches 49.5\% and \sys achieves 47\%, compared to only 18\% for the best single model. 
These results show that quorum-based early termination eliminates computational waste without sacrificing correctness. \sys matches Multi-Agent-Base accuracy within 2.5\% across all benchmarks while substantially reducing latency. The small accuracy gap reflects a fundamental design choice: \sys commits once a quorum agrees rather than waiting for all agents. For accuracy-critical applications, operators can increase $\alpha$ to require stronger agreement.
    
\subsection{Deep Dive}
\label{sec:ablation}

\begin{figure}[!t]
    \centering
    \includegraphics[width=0.48\textwidth]{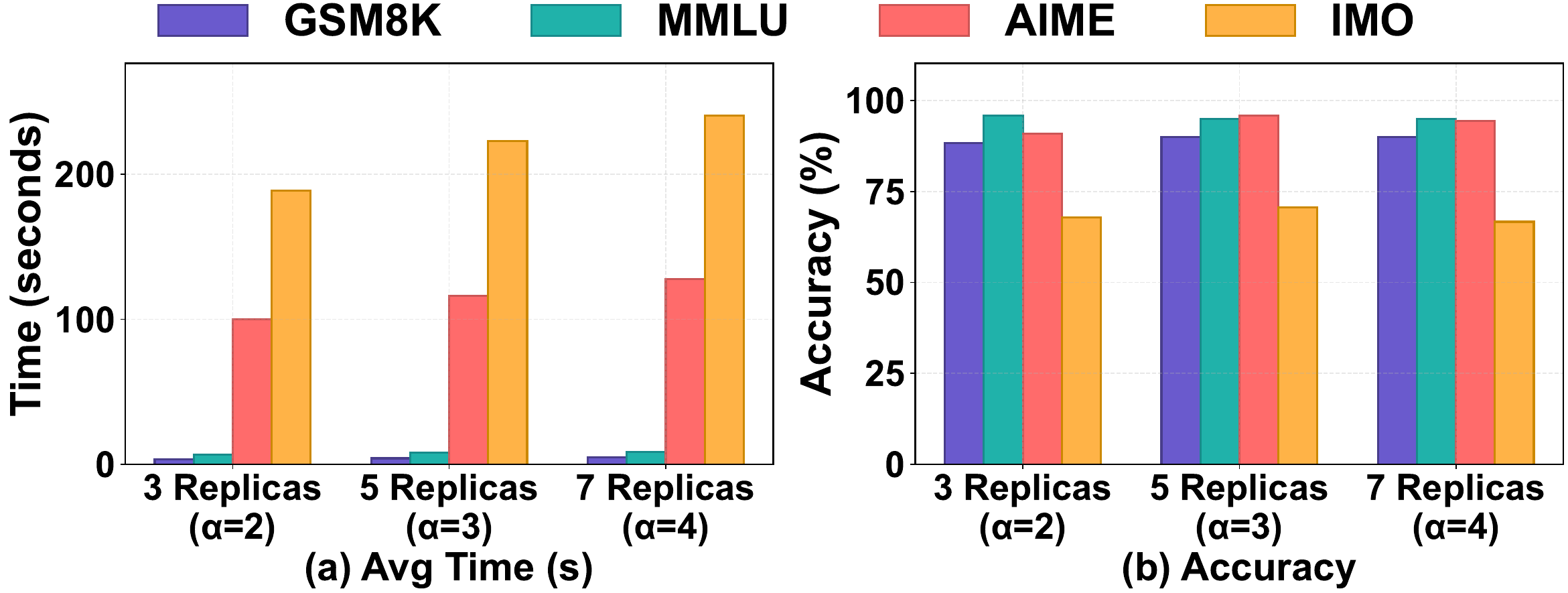}
    \caption{Impact of ensemble size on consensus performance using homogeneous replicas.}
    \label{fig:replica_count_study}
\end{figure}

\noindent
Beyond end-to-end performance, we now examine how \sys's three primary parameters (ensemble size ($N$), stability horizon ($\beta$), and quorum threshold ($\alpha$)) affect the latency-accuracy tradeoff. We vary each parameter while holding the others constant.

\para{Impact of ensemble size.}
First, we evaluate how consensus performance scales with ensemble size by varying the number of replicas from 3 to 7 using the same underlying model. We deliberately employ homogeneous Qwen3-8B replicas in this ablation to isolate the effect of ensemble size from model diversity. Using heterogeneous models would introduce a confounding factor, as performance changes could stem from either the increased replica count or the addition of models with different capabilities. By holding the model constant, we can measure how the consensus protocol's coordination overhead scales with the number of participants.
As shown in \autoref{fig:replica_count_study}, the system exhibits favorable sub-linear scaling across benchmarks of varying difficulty. For simpler benchmarks, average time per problem increases modestly from 3.6s to 4.8s on GSM8K and from 6.9s to 8.5s on MMLU when scaling from 3 to 7 replicas, representing only 33\% and 23\% overhead, respectively. The sub-linear pattern extends to challenging benchmarks: AIME increases from 99.9s to 127.7s (28\% overhead) and IMO grows from 188.8s to 240.4s (27\% overhead), indicating that coordination costs do not compound with problem difficulty. Accuracy remains stable or improves slightly across all configurations. GSM8K improves from 88.3\% to 90\%, while MMLU maintains consistent performance between 95-96\%. On AIME, accuracy peaks at 95.8\% with 5 replicas, suggesting that moderate ensemble sizes can enhance correctness on challenging problems.

\begin{figure}[!t]
    \centering
    \includegraphics[width=0.45\textwidth]{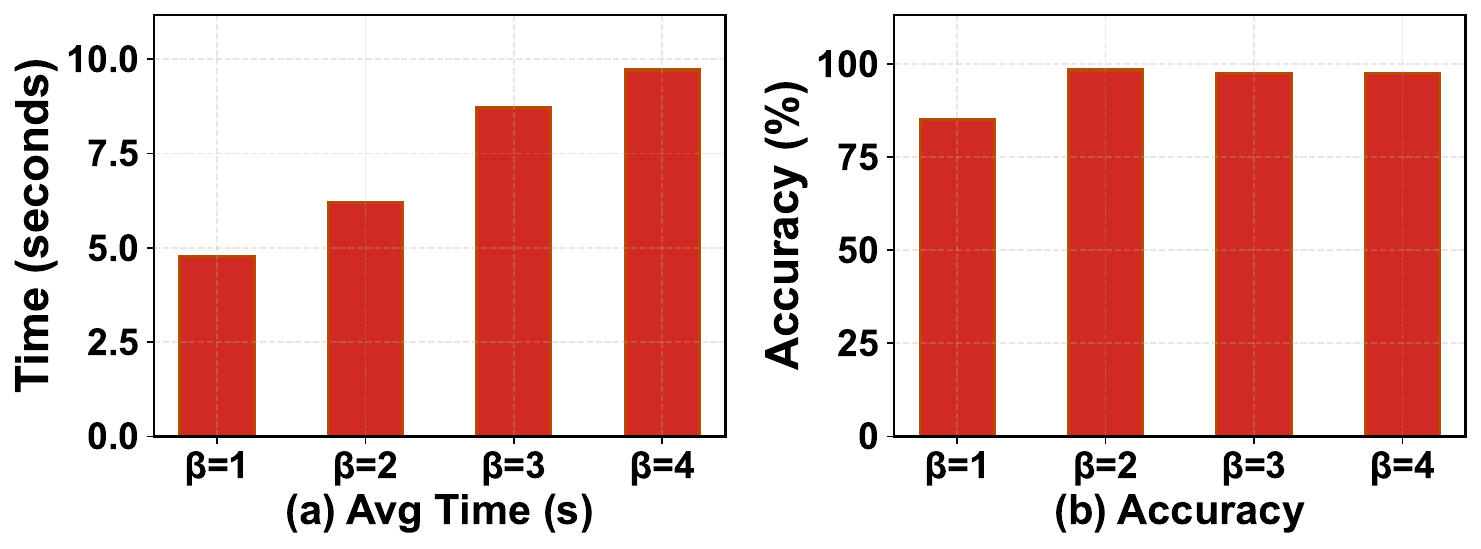}
    \caption{Impact of stability parameter $\beta$ on consensus performance with 3 models and $\alpha$=2.} 
    \label{fig:stability_study}
\end{figure}

\para{Impact of stability horizon.}
Second, we vary the stability horizon $\beta$ to understand its effect on convergence with the GSM8K dataset. \autoref{fig:stability_study}-(a) shows that average time per problem increases linearly from 4.8s at $\beta=1$ to 9.7s at $\beta=4$. However, the accuracy results in \autoref{fig:stability_study}-(b) reveal a more nuanced picture. At $\beta=1$, accuracy drops to 85\%, indicating premature termination before models reach true consensus. In contrast, $\beta=2$, $\beta=3$, and $\beta=4$ all achieve approximately 98\% accuracy, suggesting they allow sufficient time for stable agreement. Increasing $\beta$ beyond 2 provides diminishing returns. While $\beta=3$ and $\beta=4$ maintain the same accuracy as $\beta=2$, they incur substantial time penalties. We therefore select $\beta=2$ for 3 agents as our default, balancing accurate consensus with convergence efficiency.

\begin{figure}[!t]
    \centering
    \includegraphics[width=0.45\textwidth]{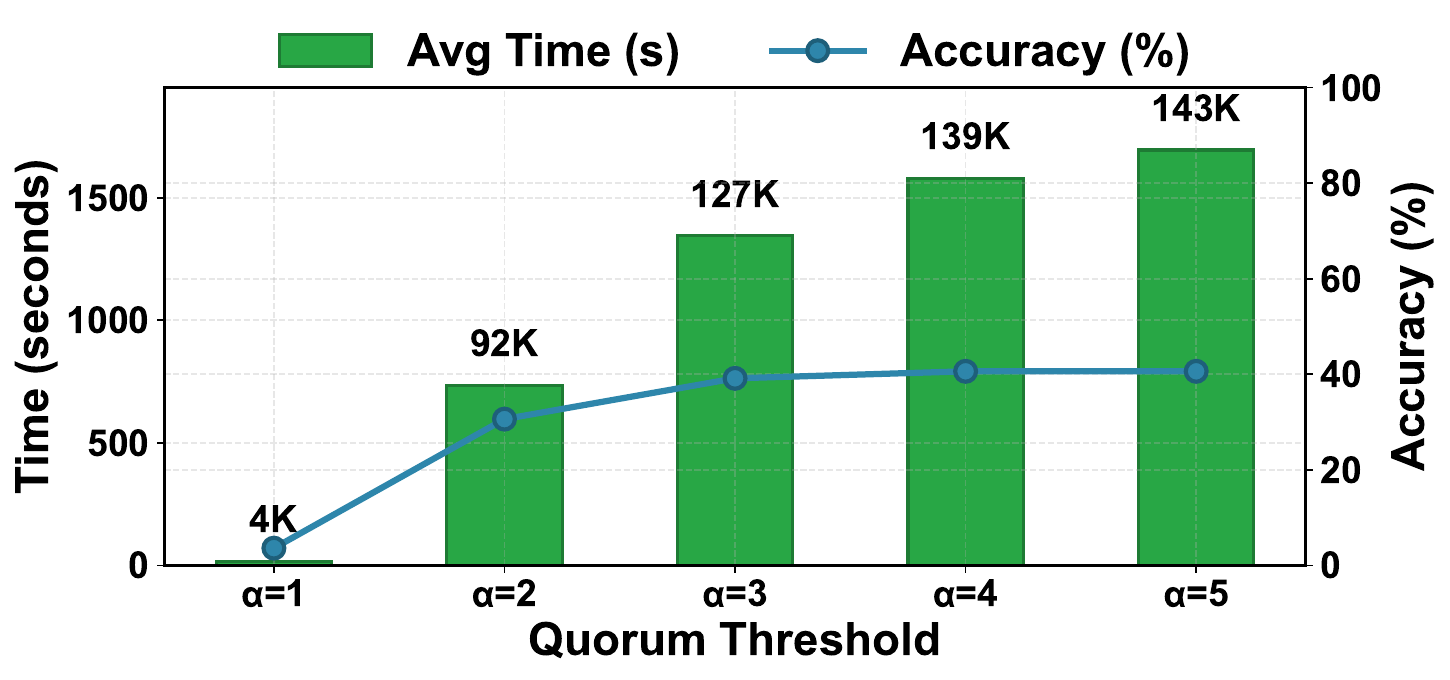}
    \caption{Impact of quorum threshold $\alpha$ on an ensemble of 5 models.} 
    \label{fig:quorum_study}
\end{figure}

    
\para{Impact of quorum threshold.}
Third, we examine the quorum threshold $\alpha$ using an ensemble of 5 open source models in \autoref{sec:exp-setup}. We use the most difficult dataset IMO. \autoref{fig:quorum_study} reveals a clear tradeoff between latency and accuracy. At $\alpha$=1, the system terminates immediately with any single answer, achieving only 17s latency but near-zero accuracy. As $\alpha$ increases, accuracy improves sharply from $\alpha$=1 to $\alpha$=3 (majority threshold), then plateaus: $\alpha$=3, $\alpha$=4, and $\alpha$=5 all achieve approximately 43\% accuracy, while latency continues to grow from 1346s to 1695s and token consumption increases from 127K to 143K. This pattern validates the classical $f+1$ of $2f+1$ quorum design: The majority threshold ($\alpha$=3 for N=5) captures the consensus signal without paying the additional cost of waiting for slower agents. Stricter thresholds provide no accuracy benefit while imposing 17-26\% additional latency overhead.

\begin{figure}[!t]
    \centering
    \includegraphics[width=\linewidth]{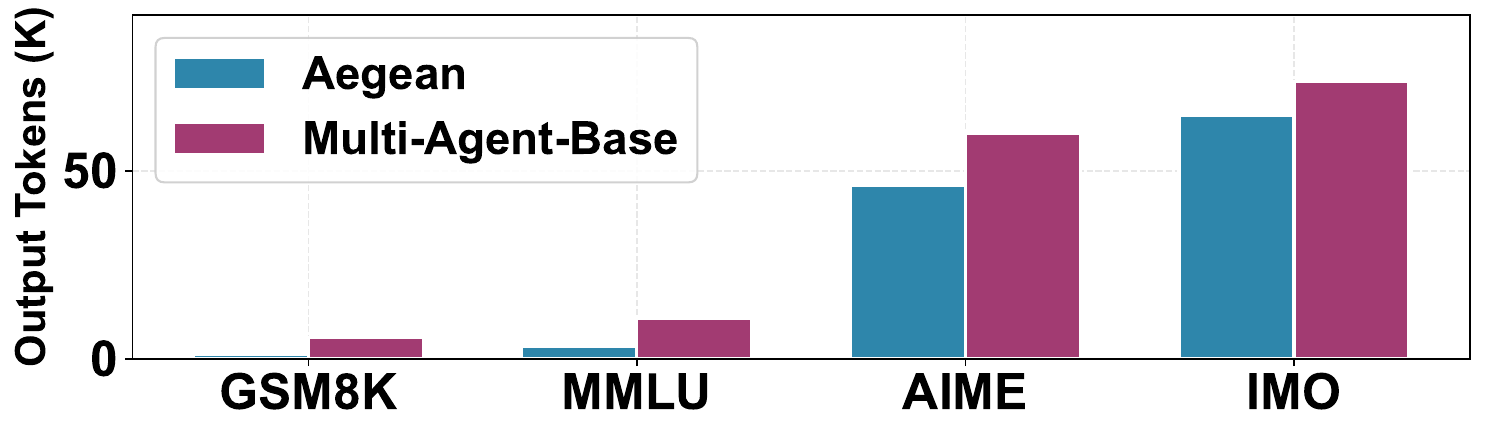}
    \caption{Average output tokens per problem for \sys and Multi-Agent-Base.}
    \label{fig:token-efficiency}
\end{figure}

\subsection{Computational Cost Reduction.}
\label{subsec:cost}
\noindent
Finally, we measure the token savings from early termination. \autoref{fig:token-efficiency} compares average output tokens across four benchmarks. On GSM8K and MMLU, \sys achieves 4.4$\times$ and 3.3$\times$ token reductions respectively, consuming only 1.3K and 3.3K tokens per problem compared to 5.7K and 10.7K for Multi-Agent-Base. These significant reductions reflect the relative simplicity of these benchmarks: Agents reach quorum quickly, allowing the system to terminate rounds early and cancel redundant generation. On more challenging benchmarks, reductions are smaller but remain consistent. AIME shows 1.3$\times$ reduction (46.0K vs 59.9K tokens) and IMO shows 1.1$\times$ reduction (64.8K vs 73.8K tokens). The correlation between problem difficulty and token consumption demonstrates that \sys naturally adapts computational expenditure to task complexity. Easy problems that achieve rapid consensus incur minimal overhead, while difficult problems receive the extended deliberation they require.

\section{Related Work}

\para{Traditional consensus protocols.} Classic consensus algorithms~\cite{liskov2012viewstamped,oki1988viewstamped,tollman2021epaxos,li2016just,ongaro2014raft,castro1999practical,lamport1998parttime} provide strong consistency guarantees. These protocols require multi-round message exchanges and leader-based architectures, creating bottlenecks when nodes are geographically dispersed or subject to variable network conditions. Paxos and Raft require majority quorums with repeated leader coordination, while Byzantine-tolerant protocols like PBFT incur quadratic communication costs that limit scalability. Optimizations such as Fast Paxos~\cite{lamport2006fast} reduce latency but rely on super-majority quorums. More fundamentally, traditional consensus is designed for deterministic state machines and is incompatible with stochastic multi-agent reasoning. \sys adapts these primitives by replacing immediate commits with stability horizons, requiring consensus to persist for $\beta$ consecutive rounds. This filters out transient states that traditional protocols cannot handle.

\para{Multi-agent serving optimizations.}
As LLM applications evolve to complex multi-agent workflows, specialized serving frameworks have emerged. Parrot~\cite{lin2024parrot} introduces Semantic Variables for application-level optimization. Autellix~\cite{luo2025autellix} applies program-level scheduling to mitigate head-of-line blocking. Kairos~\cite{chen2025kairos} and Gradientsys~\cite{song2025gradientsys} reduce latency through workflow-aware priority scheduling. However, these frameworks optimize predetermined workflows, treating multi-agent interaction as static dependency graphs where all nodes must complete. None explores consensus-aware optimizations that adapt workloads based on output convergence. \sys addresses this by embedding consensus logic into the serving layer, using streaming quorum detection to cancel redundant generation once agreement is reached.
\section{Conclusion}

\noindent
\sys represents a shift from ad-hoc multi-agent orchestration to principled consensus-based coordination. Classical consensus protocols enable unreliable processors to act as a coherent system. Similarly, \sys transforms stochastic agents into a reliable reasoning engine. We replace heuristic workflows with formal stability horizons and quorum-based locks. This design aligns computational cost with actual convergence, filters out transient agreements, and accelerates decision-making. Our results demonstrate that \sys brings significant performance benefits to multi-agent LLM systems.

\balance
\bibliographystyle{plain}
\bibliography{ref}


\end{document}